\documentclass[letterpaper,10pt,conference]{ieeeconf}
\usepackage[T1]{fontenc}
\usepackage[latin9]{inputenc}
\usepackage{color}
\usepackage{amsmath}

\usepackage{amsthm}
\usepackage{amssymb}
\usepackage{stackrel}
\usepackage{graphicx}
\usepackage{enumerate}
\usepackage[normalem]{ulem}
\usepackage{esint}

\usepackage{soul}

\makeatletter

\theoremstyle{definition}
\newtheorem{definition}{Definition}
\theoremstyle{plain}

\theoremstyle{definition}
\newtheorem{assumption}{Assumption}
\theoremstyle{plain}
\newtheorem{theorem}{Theorem}
\theoremstyle{remark}

\theoremstyle{plain}
\newtheorem{lemma}{Lemma}
\theoremstyle{plain}

\usepackage{tikz-cd}     
\usepackage{quiver}      
\usepackage{tikz}
\usetikzlibrary{calc,decorations.pathreplacing}
\usetikzlibrary{arrows.meta,decorations.pathreplacing}
\usetikzlibrary{animations} 
\usetikzlibrary[animations] 

\usepackage{todonotes}
\todostyle{keith}{color=purple!40,inline,author=KC}
\todostyle{brandon}{color=pink!40,inline,author=BF}
\todostyle{wilmer}{color=green!40,inline,author=WL}
\todostyle{any}{color=red!30,inline,author=ANY}

\usepackage{cite}
\usepackage{multirow}

\makeatother

\usepackage{babel}

\IEEEoverridecommandlockouts

\begin{document}

\title{From Local to Global: Sheaf-Theoretic Control Barrier Functions on Manifolds}

\author{Keith A. Currier, Wilmer Leal, Brandon C. Fallin, James Fairbanks, Warren E. Dixon
\thanks{Keith Currier, Wilmer Leal, Brandon Fallin, James Fairbanks, and Warren Dixon are with the Department of Mechanical and Aerospace Engineering at the University of Florida, Gainesville, FL, USA. Emails: {\texttt{\{k.currier, wleal, brandonfallin, fairbanksj, wdixon\}@ufl.edu}}. KC, WL, BF, and WD were supported by AFOSR (FA9950-19-1-0169), and ONR (N00014-21-1-2481). JF was supported by DARPA (HR0011222038) and ONR (N00014-23-1-2339). Any opinions, findings, and conclusions or recommendations are those of the authors and do not necessarily reflect the views of the sponsoring agency.}
}

\maketitle
\begin{abstract}
Control barrier functions (CBFs) provide formal safety guarantees for dynamical systems by rendering a safe set forward-invariant via pointwise-optimal feedback control. While the CBF framework is well established for dynamical systems in Euclidean space, there exist gaps for systems evolving on non-Euclidean manifolds. For systems modeled as constrained control differential inclusions, we introduce the Sheaf-CBF: a sheaf-theoretic framework for CBFs on smooth manifolds and vector bundles. The differential inclusions are formulated intrinsically on the control bundle, and analogous regularity properties in the Euclidean case are established fiberwise under local trivializations. Regularity conditions are given under which a pointwise-optimal feedback law is single-valued and continuous over local charts. Framed as a local-to-global optimization problem, sheaf-theoretic restriction and gluing compatibility conditions on overlapping charts are shown to uniquely determine a continuous global section of the control bundle, thereby guaranteeing forward invariance of the safe set. Numerical simulations on a nontrivial M{\"o}bius strip demonstrate that failure of sheaf compatibility results in violations of safety, whereas the Sheaf-CBF maintains forward invariance.  
\end{abstract}

\section{Introduction}
Control barrier functions (CBFs) are widely used in safety-critical control design because they provide formal safety guarantees through explicitly defined control input constraints, which render a set of states forward invariant \cite{Ames.Grizzle.ea2014, Ames.Coogan.ea2019}. Modeling dynamical systems as differential inclusions with safety guarantees provides a method of capturing uncertainty by allowing a range of possible trajectories while remaining safe, and is useful for robust control design for continuous, discrete, and hybrid systems \cite{Goebel.Sanfelice.ea2012, Sanfelice2021}.
The work in \cite{Isaly.Mamaghani.ea2024} investigated conditions under which feedback controllers defined by vector-valued CBFs remain feasible and continuous for continuous-time dynamics modeled as a constrained control differential inclusion. The vector-valued CBF approach was extended to hybrid  systems in \cite{Gonzalez.Sweatland.ea2025}, establishing analogous regularity conditions and safety guarantees for the pointwise-optimal feedback law. 

\begin{figure}[h]
    \centering
    \includegraphics[width=1.0\linewidth]{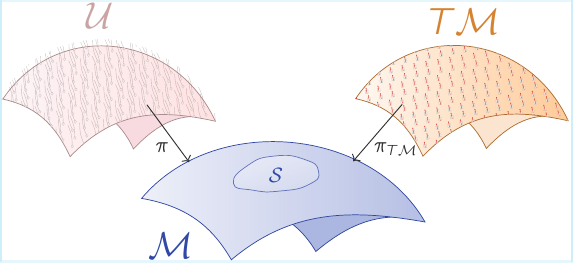}
    \caption{Control and tangent bundles over the state manifold. The control bundle $\pi:\mathcal{U}\to\mathcal{M}$ and tangent bundle $\pi_{T\mathcal{M}}:T\mathcal{M}\to\mathcal{M}$ associate each state $x\in\mathcal{M}$ with its control fiber $\mathcal{U}_x$ and tangent space $T_x\mathcal{M}$, respectively. The safe set $\mathcal{S}\subseteq\mathcal{M}$ is shown in the state manifold.}
    \label{fig:3D-Vector-Bundle-Diagram}
\end{figure}
The CBF framework is well established for systems in Euclidean space, where the state and control input spaces are modeled as a single global Cartesian product ($\mathbb{R}^n \times \mathbb{R}^m$). However, there is a substantial gap in the literature for CBFs over non-Euclidean spaces, which arise naturally in geometric control systems. Using differential geometry and geometric control formulations from \cite{Bullo.Lewis2005}, the first manifold-based CBF formulation in \cite{Wu.Sreenath2015} constructs CBFs for mechanical systems that exist on a non-Euclidean manifold like $\mathrm{SO(3)}$. Recent work in \cite{deSa.Ong.ea2025} constructs CBFs in the geometric setting by modeling nonlinear control systems as vector bundles with maps from a control input manifold into the tangent space, representing system dynamics as velocity vector fields on the tangent bundle. The two-part series in \cite{Jirwankar.Sanfelice2025a} and \cite{Jirwankar.Ochoa.ea2025b} formalize hybrid dynamical systems for a continuously differentiable manifold, recovering hybrid Lyapunov stability and general hybrid invariance. 

However, existing approaches lack a fully intrinsic formulation of constrained control differential inclusions on vector bundles. Our work addresses this gap for differential inclusions and enables a coordinate-free representation of constrained dynamics on manifolds via a local chart construction. The motivating challenge for formal safety guarantees on manifolds with no global chart naturally separates into two questions: whether safe optimal controllers can be constructed with appropriate regularity locally, and whether these local controllers are compatible enough on overlaps to define a single global feedback law. 
Sheaf theory is precisely the mathematical framework for addressing this local-to-global problem through restriction and gluing \cite{MacLane.Moerdijk.ea1994} and related sheaf-theoretic constructions have been used to assemble compatible local models into global ones in control-theoretic settings~\cite{LealNarratives2026,TemporalCellularSheaf2026}. 
Local feedback laws can be viewed as sections over an open cover of the state space, with compatibility conditions ensuring they assemble into a global section of the control bundle. This perspective allows the pointwise-optimal control law of the CBF to be formulated as a global section obtained by gluing compatible local controllers, which we formalize as the Sheaf-CBF.

The contributions of this work address these two aspects of the problem. First, we reformulate the constrained differential inclusion system as a model on the control bundle and tangent bundle. For this reframed problem, we extend regularity properties to constrained control differential inclusions on vector bundles, fiberwise, providing local chartwise regularity conditions using standard set-valued function analysis \cite{Rockafellar.Wets2009} and recent CBF motivating work \cite{Isaly.Mamaghani.ea2024, Gonzalez.Sweatland.ea2025}, and \cite{Maghenem.Sanfelice2021} on the bundle itself. With fiberwise regularity conditions, we formulate sheaf-theoretic restriction and gluing conditions to find the global section of the control bundle that yields a continuous feedback control law. The main contribution is the first sheaf-theoretic CBF framework that establishes forward invariance guarantees via restriction and gluing conditions on overlapping charts. Using local charts ensures sheaf-theoretic guarantees for nontrivial control bundles. To validate this approach, we provide simulation results for a target tracking problem on a nontrivial control bundle modeled as a M{\"o}bius strip. The simulations demonstrate that failing to incorporate proper sheaf conditions leads to safety violations, whereas the Sheaf-CBF maintains forward invariance throughout. 

\section{Notation and Preliminaries}
\subsection{General Notation}
For a set $A$, the boundary is denoted as $\partial A$, its closure as $\bar{A}$, its interior as $\mathrm{int}(A)$, and an open neighborhood of $A$ as $\mathcal{N}(A)$. The set $A$ is compact if every open cover of $A$ has a finite subcover, and is precompact if $\bar{A}$ is compact. For the sets $X,\,Y,\,Z$ and maps $f:X\to Y$ and $g:Y\to Z$, their composition  $g\circ f:X\to Z$ is defined by $(g\circ f)(x)\triangleq g(f(x))$. The identity map $\mathrm{id}_X:X\rightarrow X$ is defined by $\mathrm{id}_X(x)=x$. For a Cartesian product $X_1 \times X_2$, the projection onto the first and second factors is denoted by $\mathrm{pr}_1$ and $\mathrm{pr}_2$, respectively. The general linear group over $\mathbb{R}$, denoted by $\mathrm{GL}(n,\mathbb {R})$, is the group of all invertible $n\times n$ matrices with real entries, whose group operation is matrix multiplication. The notation $X\cong Y$ denotes an isomorphism between $X$ and $Y$. A set-valued map (multifunction) $\Psi:A\rightrightarrows B$ assigns each $x\in A$ a set $\Psi(x)\subset B$. The domain of $\Psi$ is $\mathrm{dom}(\Psi)\triangleq \{ x\in A \mid \Psi(x)\neq \emptyset\}.$
The set-valued map $\Psi$  is outer semicontinuous (osc) at $x$ if ${\mathrm{lim}_{z\rightarrow x}\,\mathrm{sup}\,\Psi(z)}\subset\Psi(x)$, inner semicontinuous (isc) at $x$ if $\liminf_{z\rightarrow x}\,\Psi(z)\supset\Psi(x)$, and is continuous at $x$ if $\Psi$ is osc and isc at $x$.

\subsection{Differential Geometry}
Introducing control differential inclusions on general manifolds requires an overview of differential geometry. For in-depth coverage of differential geometry, see works such as \cite{Lee2012} and \cite{Tu2011}, which we use to establish the subsequent definitions. Throughout, all manifolds are smooth (Hausdorff and second-countable), and all vector bundles are smooth. 

For a smooth manifold $\mathcal{M}$ (base manifold), a smooth space over $\mathcal{M}$ is a smooth projection map $\pi: E\rightarrow\mathcal{M}$ from a smooth manifold $E$ (total space). The \textit{fiber} over $p\in\mathcal{M}$ is $E_p\triangleq \pi^{-1}(p)$, where a point on the fiber is an element $e_p\in E_p$ of the total space with respect to the base manifold. A \textit{section} is a map $s:\mathcal{M}\rightarrow E$ where  $s(p)\in E_p$ for each $p\in\mathcal{M}$. The \textit{space of global sections} is defined as $\Gamma(E)=\{s:\mathcal{M}\rightarrow E\mid \pi\circ s=\mathrm{id}_{\mathcal{M}}\}$. For any open set $U\subseteq\mathcal{M}$, the \textit{restriction} of the total space is defined as  $\left. E \right|_{U}\triangleq\pi^{-1}(U)$. The local section over $U$ is a map $s:U\rightarrow \left . E \right|_{U}$ where $s(p)\in E_p$ for every $p\in U$. The \textit{space of local sections} is defined as $\Gamma(U,E)=\{s:U\rightarrow \left. E \right|_{U} \mid \pi\circ s=\mathrm{id}_{U}\}$. The \textit{tangent space} of $\mathcal{M}$ at $p\in\mathcal{M}$ is denoted as $T_p\mathcal{M}$. 

A \textit{fibered manifold} over a smooth manifold $\mathcal{M}$ is a triple $(E,\pi,\mathcal{M})$, where $E$ is a smooth manifold and $\pi$ is a smooth surjective submersion $\pi:E\rightarrow\mathcal{M}$. 
A fibered manifold $(E,\pi,\mathcal{M})$ is a
\textit{smooth fiber bundle} with model fiber $F$ if $F$ is a smooth manifold and,
for every $x\in\mathcal{M}$, there exists an open neighborhood $U$ and a diffeomorphism $\phi:\pi^{-1}(U)\to U\times F$, called a
\textit{local trivialization}, that is fiber-preserving, i.e.,
$\mathrm{pr}_1\circ\phi=\pi$ on $\pi^{-1}(U)$. The \textit{tangent bundle} is the disjoint union of tangent spaces of $\mathcal{M}$ as $T\mathcal{M}=\bigsqcup_{p\in\mathcal{M}} T_{p}\mathcal{M}$ with $\pi_{T\mathcal{M}}: T\mathcal{M}\rightarrow\mathcal{M}$. For a subset $A\subset E$, the projection is denoted $\Pi(A)\triangleq\pi(A)=\{p\in\mathcal{M}\mid  \exists\, e_p\in E_p \; \text{s.t.} \;  e_p \in A\}$. A \textit{rank-$k$ vector bundle} is a fiber bundle $(E,\pi,\mathcal{M})$ with the model fiber $F=\mathbb{R}^k$, where each fiber $E_p$ is a $k$-dimensional real vector space and each local trivialization $\phi:\pi^{-1}(U)\rightarrow U \times \mathbb{R}^k$ restricts to a linear isomorphism on each fiber.
A vector bundle is \textit{trivial} if it admits a global trivialization, i.e., the local trivialization  $\phi:\pi^{-1}(U)\rightarrow U \times \mathbb{R}^k$ extends to all $\mathcal{M}$ with $U=\mathcal{M}$. 
A \textit{chart} on a manifold $\mathcal{M}$ is a pair $(U,\phi)$, where $U\subseteq\mathcal{M}$ is the open coordinate domain and $\phi$ is a coordinate map. For a vector bundle, a \textit{bundle atlas} $\{(U_\alpha,\phi_\alpha)\}_{\alpha \in A}$ is a collection of local trivializations covering $\mathcal{M}$.

Any nontrivial bundle requires multiple local trivializations to correctly map coordinate changes across fibered elements.
Following \cite[Lem. 10.5]{Lee2012}, given a vector bundle $(E,\pi,\mathcal{M})$, and two overlapping open sets $U_\alpha\cap U_\beta\neq \emptyset$ of $\mathcal{M}$, and two local trivializations on $E$ given by $\phi_\alpha:\pi^{-1}(U_\alpha)\rightarrow U_\alpha \times \mathbb{R}^k$ and $\phi_\beta:\pi^{-1}(U_\beta)\rightarrow U_\beta \times \mathbb{R}^k$,  there exists a smooth map called a \textit{transition function} $g_{\alpha\beta}: U_\alpha\cap U_\beta\rightarrow \mathrm{GL}(k,\mathbb{R})$ that linearly maps fibered elements on the overlap. The transition function maps $\phi_\alpha \circ \phi_{\beta}^{-1}$ from $(U_\alpha\cap U_\beta) \times \mathbb{R}^k$ to itself and has the form $\phi_\alpha \circ \phi_{\beta}^{-1}(p,v)=(p,g_{\alpha\beta}(p)v)$, where $g_{\alpha\beta}(p)v$ denotes at $p$ the matrix multiplication of the $k\times k$ matrix $g_{\alpha\beta}(p)\in\mathrm{GL}(k,\mathbb{R})$ on the vector $v\in\mathbb{R}^k$.

\subsection{Euclidean Constrained Differential Inclusions}
Consider a constrained control differential inclusion $(\Phi,C_u)$ with the state $x\in\mathbb{R}^n$ and input $u\in\mathbb{R}^m$ such that 
\begin{equation}
\label{eq:Euclidean_system}
    \dot{x}\in \Phi(x,u),\quad (x,u)\in C_u,
\end{equation}
where $\Phi:\mathbb{R}^n \times\mathbb{R}^m \rightrightarrows \mathbb{R}^n$ is the set-valued flow map and $C_u\subset\mathbb{R}^n\times\mathbb{R}^m$ is the set of flows. The set of states where flow exists is defined by 
\begin{equation}
\label{eq:Euclidean_Flow_State_Set}
    \Pi(C_u)\triangleq \{x\in\mathbb{R}^n\mid \exists\, u\in \mathbb{R}^m\, \text{s.t.}\; (x,u)\in C_u\},
\end{equation}
where the set of admissible control inputs for $x\in \Pi(C_u)$ is 
\begin{equation}
\label{eq:Euclidean_Admissible_Set}
    \Psi(x) \triangleq \{ u\in \mathbb{R}^m  \mid (x,u)\in C_u\}.
\end{equation}

Given a safe set $\mathcal{S}\subset\mathbb{R}^n$, the goal is to construct a CBF candidate to guarantee the existence of control inputs satisfying the forward invariance of flows in $\mathcal{S}\subset\Pi(C_u)$. 

\begin{definition}\textcolor{black}{\cite[Def. 1]{Isaly.Mamaghani.ea2024}}
\label{def:Euclidean_CBF_Candidate}
A continuously differentiable vector-valued function 
$B:\mathbb{R}^n \rightarrow \mathbb{R}^v$ with components 
$B(x) = (B_1(x), \ldots, B_v(x))$ is a CBF candidate 
defining the safe set $\mathcal{S} \subset \Pi(C_u)$ if
\begin{equation}
    \mathcal{S} = \{x \in \Pi(C_u) : B(x) \leq 0\}.
\end{equation}
For each $i \in [v]$, define $\mathcal{S}_i \triangleq 
\{x \in \Pi(C_u) : B_i(x) \leq 0\}$ and 
$M_i \triangleq \{x \in \partial\mathcal{S} : B_i(x) = 0\}$.
Given a CBF candidate, the worst-case flow growth 
function $\Upsilon: C_u \rightarrow \mathbb{R}^v$ with 
$i^{\text{th}}$ component is denoted by
\begin{equation}
    \Upsilon_i(x, u) \triangleq 
    \sup_{f \in \Phi(x,u)} \langle \nabla B_i(x),\, f\rangle, 
    \quad \forall\, (x,u) \in C_u.
    \label{eq: Euclidean_flow_growth}
\end{equation}
\end{definition}
To enforce constraint satisfaction, a performance function $\gamma:\Pi (C_u) \to \mathbb{R}^{v}$ is introduced as a design parameter following \cite{Isaly.Mamaghani.ea2024}. Each $i\in[v]$, the component $\gamma_i:\Pi (C_u) \to \mathbb{R}$ constrains the worst-case growth of $\Upsilon_i$. 

\begin{definition}\textcolor{black}{\cite[Def. 2]{Isaly.Mamaghani.ea2024}}
\label{def: Euclidean_CBF}
A CBF candidate $B:\mathbb{R}^n \rightarrow \mathbb{R}^v$ 
defining $\mathcal{S}$ is a CBF for \eqref{eq:Euclidean_system} 
and $\mathcal{S}$ on a set $\mathcal{P} \subset \Pi(C_u)$ 
with respect to a performance function 
$\gamma:\Pi(C_u) \rightarrow \mathbb{R}^v$ if there exists 
a neighborhood $\mathcal{N}(\partial\mathcal{S})$ such that: 1) $\mathcal{N}(\partial \mathcal{S}) \subset \mathcal{P}$; 2) for each $i \in [v]$, $\gamma_i(x) \geq 0$ for all $x \in \mathcal{N}(M_i) \setminus \mathcal{S}_i$; 3) for all $x \in \mathcal{P}$,  the set of control inputs is nonempty:
    \begin{equation}
        K_{\mathrm{CBF}}(x) \triangleq \{u \in \Psi(x) : 
        \Upsilon(x,u) \leq -\gamma(x)\}
        \label{eq: Euclidean_KCBF}
    \end{equation}
\end{definition}
The safety-constrained feedback law $\kappa^*:\mathbb{R}^n \to \mathbb{R}^m$ is defined as the pointwise-optimal selection
\begin{equation}
    \kappa^*(x) \triangleq 
    {\mathrm{arg\,min}_{u \in \Psi(x)}}\; Q(x,u) \quad 
    \text{s.t.}\quad \mathcal{C}(x,u)\leq 0,
    \label{eq: Euclidean_QP}
\end{equation}
where $\mathcal{C}(x,u)\triangleq \Upsilon(x,u) +\gamma(x)$ is the constraint function, $Q: \mathcal{P} \times \mathbb{R}^m \rightarrow 
\mathbb{R}$ is a user-defined cost, typically 
$Q(x,u) = \|u - u_{\mathrm{nom}}\|^2$ for a nominal 
controller $u_{\mathrm{nom}} \in \Psi(x)$, ensuring 
$\kappa^*(x) \in \Psi(x)$ for all $x \in \mathcal{P}$.

The construction of constrained control differential inclusions on Euclidean spaces admits a natural geometric interpretation. The connection to vector bundles follows by recognizing that the product space $\mathbb{R}^n \times \mathbb{R}^m$ is canonically isomorphic to a globally trivial total space $\mathcal{U}$ \cite[Ex. 10.2]{Lee2012} where control inputs live fiberwise over the state space $\mathcal{M}=\mathbb{R}^n$, yielding the bundle $(\mathbb{R}^n\times\mathbb{R}^m,\mathrm{pr}_1,\mathbb{R}^n)\cong(\mathcal{U},\pi,\mathcal{M}).$ 
The vector bundle analog of the set-valued flow map $\Phi$ is a multifunction over the state space. Specifically, $\Phi$ maps from the product space $\mathcal{U}$ into the tangent space $T\mathcal{M}\cong \mathbb{R}^n \times \mathbb{R}^n$. The approach in \cite{deSa.Ong.ea2025} introduced this vector bundle formulation and models a control system's dynamics as a single-valued map. In this geometric setting, set-valued maps $\Psi:\mathbb{R}^n \rightrightarrows\mathbb{R}^m$ are fiber-valued multifunctions on the trivial bundle. On a general manifold $\mathcal{M}$, the input manifold $\mathcal{U}$ may be nontrivial if it admits no global trivialization. For control on a nontrivial rank-$m$ control bundle, $u_x\in\mathcal{U}_x$ denotes the intrinsic control input of the fiber $\mathcal{U}_x$ over $x$. The intrinsic element $u_x$ is coordinate independent, whereas its local coordinate representation depends on the trivialization. Under a local trivialization $\phi_{\alpha}$ over $U_{\alpha}$, the local
representative is $u_\alpha \triangleq\mathrm{pr}_{2}\circ\phi_{\alpha}(u_{x})
\in\mathbb{R}^{m}$. On an overlap $x\in U_{\alpha}\cap U_{\beta}$, the
representatives $u_{\alpha},u_{\beta}\in\mathbb{R}^{m}$ describe the same
intrinsic input $u_{x}$ when $u_{\alpha}=g_{\alpha\beta}(x)\,u_{\beta}$, where $g_{\alpha\beta}(x)\in\mathrm{GL}(m,\mathbb{R})$ is the transition function. A standard example is the M\"obius strip, a nontrivial real line bundle over a circle ($S^1$) which has no global trivialization and  transitions occur as sign changes of the fiber coordinate across chart overlaps.

\subsection{Sheaf Compatibility}
\label{sec:Sheaf Compatibility}
A feedback law is a section of the control bundle assigning a control input to every state in its domain.  Therefore, pointwise compatibility under bundle transition maps must extend to compatibility of local feedback laws on chart overlaps.
If local feedback laws fail to represent the same intrinsic section on overlaps, a global feedback law cannot be assembled. The M\"obius strip example in Section \ref{sec: Simulated Results} illustrates how neglecting local compatibility can result in a loss of safety. Sheaf theory provides the formalism for synthesizing local objects that agree on overlaps into a unique global object, specifically through restriction and gluing~\cite{MacLane.Moerdijk.ea1994,Lee2024}.  Accordingly, we organize continuous feedback laws as local sections and introduce the corresponding presheaf and sheaf structure to characterize when local controllers uniquely determine a global feedback law.
 
Given a rank-$k$ vector bundle $(\mathcal{U},\pi,\mathcal{M})$ and an open set $U\subseteq\mathcal{M}$, a local section of the restricted bundle $\mathcal{U}\big|_U$ is a map $\kappa:U\to\mathcal{U}\big|_U$ satisfying $\pi\circ\kappa=\mathrm{id}_U$, and the space of local sections is given by $\Gamma(U,\mathcal{U})=\{\kappa:U\to\mathcal{U}\big|_U\mid\pi\circ\kappa=\mathrm{id}_U\}$. Restricting to continuous sections, define $\mathcal{F}(U)\triangleq\{\kappa\in\Gamma(U,\mathcal{U})\mid\kappa\text{ is continuous}\}$.
Let $\mathcal{O}(\mathcal{M})$ denote the category of open subsets of $\mathcal{M}$ ordered by inclusion, and let $\mathbf{Set}$ denote the category of sets.  Consider the assignment $U \mapsto \mathcal{F}(U)$ on objects and, for each inclusion $V \subseteq U$, the assignment $(V \subseteq U) \mapsto \rho_V^U$ on morphisms, where $\rho_V^U:\mathcal{F}(U)\to\mathcal{F}(V)$ is the restriction map defined by $\rho_V^U(\kappa)=\kappa|_V$. These assignments define a presheaf
\begin{equation}
\label{eq:presheaf-sections}
    \mathcal{F}:\mathcal{O}(\mathcal{M})^{\mathrm{op}}\to \mathbf{Set}.
\end{equation}
To establish that $\mathcal{F}$ is a sheaf, it remains to verify the gluing property for compatible local sections. For any open cover $\{U_\alpha\}_{\alpha\in A}$ of $U$, a collection $\{\kappa_\alpha\in\mathcal{F}(U_\alpha)\}_{\alpha\in A}$ that agrees on overlaps satisfying
\begin{equation}
\label{eq: Restriction Compatibility condition}
    \kappa_\alpha\big|_{U_{\alpha} \cap U_\beta} = \kappa_\beta\big|_{U_{\alpha} \cap U_\beta},\quad \forall\,\alpha,\beta\in A,
\end{equation}
uniquely determines a continuous section $\kappa \in \mathcal{F}(U)$ such that $\kappa\big|_{U_\alpha}= \kappa_\alpha$. Hence, $\mathcal{F}$ is a sheaf. Compatibility ensures that the pointwise definition $\kappa(x)=\kappa_\alpha(x)$ for $x\in U_\alpha$ is independent of $\alpha$, while the continuity of $\kappa$ follows locally from the continuity of $\kappa_\alpha$.  Thus, compatible continuous local feedback laws glue uniquely to a global continuous feedback law.  In the CBF setting, however, we are interested in continuous local sections satisfying the CBF constraints $\kappa(x)\in K_{\mathrm{CBF}}(x)$. For each open $U\subseteq\Pi(C_u)$, define 
\begin{equation} \label{eq: sheaf_KCBF}
\mathcal{K}_{\mathrm{CBF}}(U) \triangleq \left\{ \kappa\in\mathcal{F}(U) \mid \kappa(x)\in K_{\mathrm{CBF}}(x),\ \forall x\in U \right\}. 
\end{equation}  
Since membership in $K_{\mathrm{CBF}}(x)$ is a pointwise condition, a compatible family of sections $\{\kappa_\alpha\in\mathcal{K}_{\mathrm{CBF}}(U_\alpha)\}_{\alpha\in A}$ glues through the sheaf $\mathcal{F}$ to a unique continuous section $\kappa\in\mathcal{K}_{\mathrm{CBF}}(U)$. Thus, the difficulty is not whether compatible sections can be glued, but whether the pointwise-optimal controllers constructed on local trivializations define such sections. For a general control bundle, the admissible set $K_\mathrm{CBF}(x)$ varies fiberwise over the state manifold, so the regularity required for the feedback law must be established over local trivializations and related intrinsically to the bundle. Local controllers must also be compatible under transition maps. These requirements motivate Section~\ref{sec: Fiberwise_Regularity} and \ref{sec: Sheaf Compatibility Conditions}, respectively: establishing local fiberwise regularity and compatible gluing of local controllers.

\section{Problem Formulation} 
\label{sec: Problem Formulation}
\subsection{System Dynamics}
Let $(\mathcal{U},\pi,\mathcal{M})$ be a rank-$m$ control bundle over the state manifold $\mathcal{M}$, where $\mathcal{U}$ is the total control space, $\pi:\mathcal{U}\to \mathcal{M}$ is the bundle projection and $\mathcal{U}_x\triangleq \pi^{-1}(x)$ is the control fiber over $x\in\mathcal{M}$.
A constrained control differential inclusion $(\Phi,C_u)$ with state $x\in\mathcal{M}$ and input on a fiber $u_x\in \mathcal{U}_x$ is 
\begin{equation}
    \dot{x} \in \Phi(u_x),\quad u_x\in C_u, \label{eq:System_1}
\end{equation}
where $\Phi:\mathcal{U}\rightrightarrows T\mathcal{M}$ is a fiber-valued multifunction that associates every admissible input on the fiber $u_x\in\mathcal{U}_x$ with admissible velocities $\Phi(u_x)\subseteq T_x\mathcal{M}$, and $C_u\subset\mathcal U$ is the set of admissible inputs admitting flow.
The differential inclusion is expressed as a fiberwise Minkowski sum 
\begin{equation}
    \Phi(u_x) = F(x)+ G_xu_x \triangleq \{v_x+ G_xu_x\mid v_x\in F(x)\},
    \label{eq:System_2}
\end{equation}
where  $F:\mathcal{M}\rightrightarrows T\mathcal{M}$ is a fiber-valued multifunction that associates points on the state space to a set of the tangent bundle such that $F(x)\subseteq T_x\mathcal{M}$, and $G_x:\mathcal{U}_x \to T_x\mathcal{M}$ is the fiberwise linear map into the tangent space. The projection of $C_u \subset\mathcal{U}$ onto $\mathcal{M}$ denotes the allowable flow set as
\begin{equation}
    \Pi(C_u) \triangleq \{ x\in\mathcal{M}\mid \exists\; u_x\in \mathcal{U}_x\; \text{ with }\; u_x\in C_u \}.
\end{equation}
For each state $x\in\Pi(C_u)$, the set of admissible control inputs on the fiber is denoted as 
\begin{equation}
    \Psi(x) \triangleq \{ u_x\in C_u  \mid u_x\in \mathcal{U}_x\}=C_u\cap \mathcal{U}_x.
\end{equation}

\subsection{Control Barrier Functions}
The safe set $\mathcal{S} \subset \mathcal {M}$ is a regular domain \cite[Prop.~5.47]{Lee2012}, i.e., a properly embedded codimension-0 submanifold of $\mathcal{M}$ whose boundary $\partial\mathcal{S}$ is of codimension-1 with respect to $\mathcal{M}$. 
To characterize $\mathcal{S}$, we adopt the CBF candidate from Definition \ref{def:Euclidean_CBF_Candidate} on the state space manifold $B:\mathcal{M}\rightarrow\mathbb{R}^v$.

\begin{definition}\label{def:CBF_on_a_bundle}
A continuously differentiable CBF candidate $B:\mathcal{M} \to \mathbb{R}^{v}$ defining the safe set $\mathcal{S}$ is a CBF for \eqref{eq:System_1} and $\mathcal{S}$ on a set $\mathcal{P} \subset \mathcal{M}$ with respect to a performance function $\gamma$ if there exists a neighborhood $\mathcal{N}(\partial\mathcal{S})$ of the boundary of $\mathcal{S}$, such that: 1) $\mathcal{N}(\partial\mathcal{S}) \subset \mathcal{P}$; 2) for each $i \in [v]$, $\gamma_{i}(x) \ge 0$ for all $x \in \mathcal{N}(M_{i}) \setminus \mathcal{S}_{i}$; 3) for all $x \in \mathcal{P}$, the set
\begin{equation}
\label{eq: K_CBF}
    K_{\mathrm{CBF}}(x) \triangleq \{\, u_x \in \Psi(x) : \Upsilon(u_x) \le -\gamma(x) \,\}
\end{equation}
is nonempty, where $\Upsilon:C_u \to \mathbb{R}^{v}$ denotes the worst-case directional growth of $B$ along \eqref{eq:System_1}. For each $i \in [v]$, the component $\Upsilon_{i}:C_u \to \mathbb{R}$ is given by
\begin{equation}
        \Upsilon_{i}(u_x) \triangleq  \underset{\dot{x}\in\Phi(u_{x})}{\mathrm{sup}}\left\langle d B_{i}\left(x\right),\dot{x}\right\rangle.
\end{equation}
\end{definition}

The fiber-valued multifunction $K_{\mathrm{CBF}}:\Pi(C_u)\rightrightarrows\mathcal{U}$ assigns to each state $x\in \mathcal{P}$ a set of viable control inputs $K_{\mathrm{CBF}}(x)\subseteq \mathcal{U}_x$, each of which guarantees that solutions remain in $\mathcal{S}$ for all time. 
When $\Phi$ is nonempty and locally precompact for each $u_x\in C_u$, $\Upsilon_i(u_x)<\infty$ is well defined for $i\in [v]$ such that flows are contained. Accordingly, the following regularity assumption is imposed. 

\begin{assumption}
The fiber-valued multifunction $\Phi :\mathcal{U} \rightrightarrows T\mathcal{M}$ is locally precompact relative to $C_u$ and $C_u \subseteq \mathrm{dom}(\Phi)$. 
\end{assumption}
The closed-loop system $(\Phi_{\mathrm{cl}},C_{\kappa})$ defined by $(\Phi,C_u)$ in \eqref{eq:System_1} and with a feedback controller $\kappa:\mathcal{M} \to \mathcal{U}$ is
\begin{equation}
\label{eq: Closed_Loop_System}
    \dot{x}\in\Phi_{\mathrm{cl}}(\kappa(x)), \quad x\in \Pi (C_{\kappa}),
\end{equation}
where $C_\kappa \subset\mathcal{U}$ is the closed-loop flow set under $\kappa$, $\Pi (C_\kappa)\triangleq \{ x\in \mathcal{M} \mid \kappa(x)\in \Psi(x)  \}$ are the states where $\kappa$ admits flow, and $\Phi_{\mathrm{cl}}=F(x)+ G_x \kappa(x)$ is the closed-loop fiber-valued multifunction under the selection of $\kappa$, noting $\Phi_{\mathrm{cl}}\triangleq \emptyset$ if $x\notin \Pi (C_{\kappa})$. Next, we define notions of forward invariance for the constrained differential inclusion introduced in \eqref{eq: Closed_Loop_System}. 
\begin{definition}{\cite[Def. 4]{Isaly.Mamaghani.ea2024}}
    \label{def: Forward-Invariance-Euclidean}
    A set $\mathcal{S}\subset \Pi (C_\kappa)$ is forward pre-invariant for $(\Phi_{\mathrm{cl}}, C_\kappa)$ if for each $x(t_0)\in\mathcal{S}$ and each maximal solution $x$ starting from $x(t_0)$, $x(t)\in \mathcal{S}$ for all $t\in \mathrm{dom}(x)$. The set $\mathcal{S}$ is forward invariant for $(\Phi_{\mathrm{cl}}, C_\kappa)$ if it is forward pre-invariant and, for each $x(t_0)\in \mathcal{S}$, every maximal solution $x$ starting from $x(t_0)$ is complete. 
\end{definition}
As in the Euclidean case, the goal for CBFs modeled on bundles is to design a feedback controller $\kappa$ such that $\kappa(x) \in \Psi(x)$ for all $x \in \mathcal{M}$, and the safe set $\mathcal{S} \subset \mathcal{M}$ is forward invariant in the sense of Definition \ref{def: Forward-Invariance-Euclidean}, under the closed-loop system in \eqref{eq: Closed_Loop_System}. Reframed for vector bundles, a control input $\kappa$ in ${C_\kappa \subset\mathcal{U}}$ is selected at each state $x\in\mathcal{M}$. We define the feedback law  
$\kappa^{*}:\mathcal{M}\rightarrow\mathcal{U}$ as 
\begin{align}
\kappa^{*}\left(x\right) & \triangleq {\arg\min}_{u_x\in\Psi(x)}\;Q\left(u_x\right)\label{eq:kappa star}
\quad 
    \text{s.t.}\quad \mathcal{C}(u_x)\leq0, 
\end{align}
where $Q:\mathcal{P} \times \mathcal{U}_x \rightarrow\mathbb{R}$ is a user-defined objective function, typically as $Q(u_x)=\left\Vert u_x-u_{\mathrm{nom}}\right\Vert_{\mathcal{U}} ^{2}$, for a nominal control law $u_{\mathrm{nom}}\in\Psi$, where the norm $\left\Vert \cdot\right\Vert_{\mathcal{U}}$ is given by the fiber metric on $\mathcal{U}$. By choosing $\mathcal{C}(u_x) \triangleq \Upsilon(u_x) + \gamma(x)$, every optimizer satisfies $\kappa^*(x)\in K_{\mathrm{CBF}}(x)$ for all $x\in \mathcal{P}$. 
However, whether $\kappa^*$ defines the global section $\kappa^\star \in \Gamma(\mathcal{U})$ depends on the following: $K_\mathrm{CBF}(x) \neq \emptyset$, single-valuedness and continuity on $\Pi(C_u)$, and sheaf gluing conditions. 

\section{Sheaf Control Barrier Functions}

\subsection{Fiberwise Regularity}
\label{sec: Fiberwise_Regularity}
Ensuring the intrinsic feedback law $\kappa^*$ is feasible, single-valued, and continuous requires extending the Euclidean regularity conditions in \cite{Isaly.Mamaghani.ea2024} and \cite{Gonzalez.Sweatland.ea2025} into similar intrinsic properties of the control bundle. Over local trivializations, the map $\phi$ restricts intrinsic properties locally, and its inverse $\phi^{-1}$ lifts local properties back to the bundle. The following section uses these maps to construct a continuous fiber-valued multifunction $K_{\mathrm{CBF}}(x)$ in \eqref{eq: K_CBF} and controller $\kappa^*$ in \eqref{eq:kappa star}.
For local convexity, we define convexity over fibers. 
\begin{definition}
    $C_u \subset\mathcal{U} $  is fiberwise convex if $\Psi(x)\triangleq C_u\cap\,\mathcal{U}_x$ is a convex set in $\mathcal{U}_x$ for every $x\in\Pi(C_u)$. 
\end{definition}

The following lemma constructs local trivializations that yield fiberwise regularity for continuity, fiberwise convexity, and local precompactness for rank-$k$ vector bundles. 

\begin{lemma}
\label{lem:Local_invariants}
    Let $(E,\pi,\mathcal{M})$ be a rank-$k$ vector bundle, and let $U\subset\mathcal{M}$ be open. Let $\{(U_\alpha,\phi_\alpha)\}_{\alpha\in A}$ be a bundle atlas of $E$ where $\phi_\alpha:\pi^{-1}(U_\alpha)\rightarrow U_{\alpha}\times\mathbb{R}^k$ is a local trivialization with transition functions $g_{\alpha\beta}:U_\alpha\cap U_\beta \rightarrow \mathrm{GL}(k,\mathbb{R})$. Let $\Psi:U\rightrightarrows \left. E \right|_{U}$ be a fiber-valued multifunction such that $\Psi(x)\subset E_x$ for $x\in U$, and for each $\alpha\in A$, let $\Psi_\alpha : U_\alpha\cap U\rightrightarrows\mathbb{R}^k$ denote the local representative of $\Psi$ defined by $\Psi_\alpha(x) \triangleq \mathrm{pr}_2 \circ \phi_\alpha(\Psi(x))$. The following conditions hold for one (equivalently every) $\alpha$ with $x\in U_\alpha$: 
    \begin{enumerate}[(i)]
        \item $\Psi$ is osc, isc, or continuous at $x\in U$ if and only if  $\Psi_\alpha$ is osc, isc, or continuous at $x$, respectively;
        \item $\Psi$ is fiberwise convex for every $x\in U$ if and only if  $\Psi_\alpha(x)\subset \mathbb{R}^k$ is convex; 
        \item $\Psi$ is locally precompact at $x\in U$, if and only if  $\Psi_\alpha$ is locally bounded at $x$;
    \end{enumerate}
\end{lemma}

\begin{proof}
See Appendix A. \end{proof}

For the feedback law to be well-defined, we present the following lemma adapted from \cite[Lem.~4.1]{Gonzalez.Sweatland.ea2025} to establish the local chart regularity properties.

\begin{lemma}
 \label{lem: Continuity locally for fiber-valued multifunctions} 
    Let $\left(\mathcal{U},\pi,\mathcal{M}\right)$ be a rank-$m$ control
bundle. Let $\{(U_{\alpha},\phi_{\alpha})\}_{\alpha\in A}$ be a bundle
atlas, where $U_{\alpha}\subset\mathcal{M}$ is a trivializing neighborhood
and $\phi_{\alpha}:\left.\mathcal{U}\right|_{U_{\alpha}}\to U_{\alpha}\times\mathbb{R}^{m}$
is a local trivialization. The local representative of the admissible
flow set is defined by $C_{u,\alpha}\triangleq\phi_{\alpha}(C_{u}\cap\left.\mathcal{U}\right|_{U_{\alpha}})\subset\Pi(C_{u})_{\alpha}\times\mathbb{R}^{m}$, where $\Pi(C_{u})_{\alpha}\triangleq U_{\alpha}\cap\Pi(C_{u})\subset\mathbb{R}^{n}$.
For each $x\in\Pi(C_{u})_{\alpha}$, let  $\Psi_{\alpha}(x)\triangleq\{u\in\mathbb{R}^{m}\mid(x,u)\in C_{u,\alpha}\}$. The local representative of $K_{\mathrm{CBF}}$ is a fiber-valued multifunction $K_{\mathrm{CBF},\alpha}:\Pi(C_u)_\alpha \rightrightarrows \mathbb{R}^m$ defined by
\begin{equation}
    K_{\mathrm{CBF},\alpha}(x)\triangleq\{u\in \Psi_{\alpha}(x) \mid \eta_{i,\alpha}(x,u)\leq0,\,\forall i\in[v]\},
\end{equation}
where $\eta_{i,\alpha}(x,u)\triangleq \Upsilon_{i,\alpha}(x,u)+\gamma_{i,\alpha}(x)$. Suppose that the following conditions hold for each $\alpha \in A$ and every $x\in\Pi(C_u)_\alpha$: the set $C_{u,\alpha}$ is closed and fiberwise convex, for each $i\in[v]$, $(x,u)\mapsto \eta_{i,\alpha}(x,u)$ is continuous, $u\mapsto \eta_{i,\alpha}(x,u)$ is convex, and there exists $u\in \Psi_{\alpha}(x)$ such that $\eta_{i,\alpha}(x,u)<0$. Then $K_{\mathrm{CBF},\alpha}(x)$ is continuous on $\Pi(C_u)_\alpha$. Further, by Lemma \ref{lem:Local_invariants}, the intrinsic fiber-valued multifunction $K_{\mathrm{CBF}}(x)$ is continuous on $\Pi(C_u)$.
\end{lemma}
\begin{proof}
Let $x\in\Pi(C_{u})$ be fixed. Since $\{U_{\alpha}\}_{\alpha\in A}$
covers $\Pi(C_{u})$, there exists an $\alpha\in A$ such that $x\in U_{\alpha}$.
Under the stated closedness, convexity, continuity, and strict-feasibility conditions, the hypotheses of \cite[Lem.~4.1]{Gonzalez.Sweatland.ea2025} are satisfied, so $K_{\mathrm{CBF},\alpha}(x)$ is osc, isc and therefore continuous at every $x\in\Pi(C_{u})_{\alpha}$ by \cite[Def. 5.4]{Rockafellar.Wets2009}. Since $x\in\Pi(C_{u})$ was arbitrary, $K_{\mathrm{CBF},\alpha}(x)$ is continuous on $\Pi(C_u)_\alpha$, and by Lemma \ref{lem:Local_invariants}, continuity is invariant under local trivialization, so $K_{\mathrm{CBF}}(x)$ is continuous on $\Pi(C_u)$.
\end{proof}

Lemma \ref{lem: Continuity locally for fiber-valued multifunctions} recovers conditions where $K_\mathrm{CBF}(x)\subset \mathcal{U}$ is continuous intrinsically. 
The local representative of the control law $\kappa_\alpha^*:\Pi(C_u)_\alpha \rightarrow\mathbb{R}^m$ in \eqref{eq:kappa star} is defined as  
\begin{equation}
    \kappa_\alpha^*(x) \triangleq 
    {\mathrm{arg\,min}_{u \in K_{\mathrm{CBF},\alpha}(x)}}\; Q_\alpha(x,u),
    \label{eq: local kappa star}
\end{equation}
where $Q_\alpha: C_{u,\alpha}\rightarrow \mathbb{R}$   denotes the local representative of the cost function under the trivialization $\phi_\alpha$. 
The following theorem is a modified vector-bundle analogue of \cite[Thm. 4.3]{Gonzalez.Sweatland.ea2025} for the single-valuedness and continuity of  $\kappa_\alpha^{*}$.

\begin{theorem}
\label{thm:Local_CBF}
    Let $U_{\alpha} \subset\mathcal{M}$ be a trivializing neighborhood and let $\phi_\alpha:\left. \mathcal{U} \right|_{U_\alpha} \rightarrow U_\alpha \times \mathbb{R}^m$ be a local trivialization. The local representative of the admissible flow set is defined by $C_{u,\alpha} \triangleq \phi_\alpha (C_u \cap \left. \mathcal{U} \right|_{U_\alpha}) \subset \Pi(C_u)_\alpha \times \mathbb{R}^m$, where $\Pi(C_u)_\alpha \triangleq U_\alpha \cap \Pi(C_u) \subset\mathbb{R}^n$. For each $x\in\Pi(C_u)_\alpha$, let 
    $\Psi_\alpha(x)\triangleq\{u\in\mathbb{R}^m\mid (x,u)\in C_{u,\alpha} \}$. Let $\Upsilon_\alpha:C_{u,\alpha}\rightarrow\mathbb{R}^v$, $\gamma_\alpha:\Pi(C_u)_\alpha\rightarrow \mathbb{R}^v$, and $Q_\alpha:C_{u,\alpha}\rightarrow\mathbb{R}$ denote the local representatives of the worst-case growth function, performance function, and objective function, respectively, under the trivialization. Suppose that $C_u$ is closed in $\mathcal{U}$ and $\Psi(x)$ is locally precompact relative to $\Pi(C_u)$ and for $i\in[v]$, the following conditions hold 
\begin{enumerate}
    \item $(x,u) \mapsto \Upsilon_{i,\alpha}(x,u)+\gamma_{i,\alpha}(x)$ is continuous on $C_{u,\alpha}$;
    \item for every $x\in \Pi(C_u)_\alpha$, $u\mapsto \Upsilon_{i,\alpha} (x,u) + \gamma_{i,\alpha}(x)$ is convex on $K_{\mathrm{CBF},\alpha}(x)$ and $\Psi_\alpha(x)$ is convex in $\mathbb{R}^m$;
    \item there exists $u\in \Psi_\alpha(x)$ such that $\Upsilon_{i,\alpha}(x,u)+\gamma_{i,\alpha}(x)<0$ for every $x\in\Pi(C_u)_\alpha$; 
    \item the objective function $Q_\alpha:C_{u,\alpha}\rightarrow \mathbb{R}$  
    is continuous and, for every $x\in \Pi(C_u)_\alpha$, $u \mapsto Q_\alpha(x,u)$ is strictly convex on $K_{\mathrm{CBF},\alpha}(x)$. 
\end{enumerate}
Under the stated conditions, the local representative of $\kappa_\alpha^{*}$ in \eqref{eq: local kappa star} is single-valued and continuous on $\Pi(C_u)_\alpha$. 
\end{theorem}

\begin{proof}
The flow set $C_u\subset\mathcal{U}$ is closed and the coordinate map $\phi_\alpha$ is a homeomorphism, so the local representative $C_{u,\alpha}$ is closed in $\Pi(C_u)_\alpha \times \mathbb{R}^m$. 
Let $x\in \Pi(C_u)_\alpha$ be fixed and consider a sequence of $\{u_n \}_{n\in\mathbb{N}} \subset \Psi_\alpha(x) \subset \mathbb{R}^m$ with $(x,u_n)$ converging to $ (x,u)$ with $ u\in\mathbb{R}^m$. Since $u_n\in C_{u,\alpha}$ and $C_{u,\alpha}$ is closed, $(x,u)\in C_{u,\alpha}$. Hence, $u\in  \Psi_\alpha(x)$ and $\Psi_\alpha(x)$ is closed. With $\Psi_\alpha(x)$ being closed for every $x\in\Pi(C_u)_\alpha$ and given conditions 1)-3), the assumptions in Lemma \ref{lem: Continuity locally for fiber-valued multifunctions} are satisfied, and $K_{\mathrm{CBF},\alpha}(x)\subset\Psi_\alpha(x)$ is continuous. Further, $K_{\mathrm{CBF},\alpha}(x)\subset\Psi_\alpha(x)$ is locally bounded relative to $\Pi(C_u)_\alpha$ by assumption. From condition 4), for each $x\in\Pi(C_u)_\alpha$, $u\mapsto Q_\alpha(x,u)$ is continuous and strictly convex in $\{x\} \times K_{\mathrm{CBF},\alpha}(x) \subset \{x\} \times \mathbb{R}^m$. Since $Q_\alpha$ is continuous and $K_{\mathrm{CBF,}\alpha}$ is continuous and locally bounded at $x$, $\kappa_{\alpha}^*(x)$ is osc at $x$. The strict convexity of $u\mapsto Q_\alpha(x,u)$ on $K_{\mathrm{CBF},\alpha}(x)$ ensures uniqueness of the minimizer, so $\kappa_{\alpha}^*(x)$ is single-valued. By \cite[Cor. 5.20]{Rockafellar.Wets2009}, a locally bounded, osc single-valued map is continuous, so $\kappa_{\alpha}^*$ is continuous.
\end{proof}

With formalized fiberwise regularity, we establish the following assumption and lemma to regulate the fiberwise regularity conditions for the open-loop and closed-loop systems in \eqref{eq:System_1} and \eqref{eq: Closed_Loop_System}, respectively.

\begin{assumption}
\label{asm: Differential_inclusion_Axton}
The fiber-valued multifunction $\Phi:\mathcal{U} \rightrightarrows T\mathcal{M}$ is osc and locally precompact relative to $C_u$ with nonempty, fiberwise convex values and $C_u$ is closed in $\mathcal{U}$. 
\end{assumption}

\begin{lemma} For the closed-loop system $(\Phi_\mathrm{cl},C_\kappa)$ in \eqref{eq: Closed_Loop_System}, suppose $\kappa: \mathcal{M} \rightarrow \mathcal{U}$ and $\Pi (C_\kappa)=\{x\in\mathcal{M} \mid \kappa(x)\in \Psi(x) \}$ are such that $\kappa$ is continuous on $\Pi (C_\kappa)$. If Assumption \ref{asm: Differential_inclusion_Axton} holds, then the closed-loop fiber-valued multifunction  $\Phi_{\mathrm{cl}}:\mathcal{U} \rightrightarrows T\mathcal{M}$ is locally precompact relative to $C_\kappa$, osc, nonempty, and fiberwise convex for $C_\kappa$. Additionally, if Assumption \ref{asm: Differential_inclusion_Axton} holds, then $ C_\kappa$ is closed in $\mathcal{U}$. 
\end{lemma}

\begin{proof}
Let $\{(U_\alpha,\phi_\alpha)\}_{\alpha\in A}$ be a bundle atlas for $\mathcal{U}$, fix any chart $(U_\alpha,\phi_\alpha)$, and let $\Pi(C_\kappa)_\alpha \triangleq U_\alpha\cap \Pi(C_\kappa)$. The local representative of $\kappa$ is given by $\kappa^*_\alpha\triangleq\mathrm{pr}_2 \circ \phi_\alpha \circ \kappa|_{\Pi(C_\kappa)_\alpha}$ where $\kappa^*_\alpha:\Pi(C_\kappa)_\alpha\to\mathbb{R}^m$. Let $C_u$ be closed in $\mathcal{U}$. Since $\phi_\alpha$ is a homeomorphism, $\kappa^*_\alpha$ is continuous on $\Pi(C_\kappa)_\alpha$ and $C_{u,\alpha}$ is closed in $\Pi(C_u)_\alpha \times \mathbb{R}^m$.  Further, by Lemma \ref{lem:Local_invariants}, the open-loop dynamics satisfying Assumption \ref{asm: Differential_inclusion_Axton} restricts to the local representative: $\Phi_\alpha$ is osc, locally bounded, nonempty, and convex on $C_{u,\alpha} \subset \Pi(C_u)_\alpha \times \mathbb{R}^m$. By \cite[Lem. ~2]{Isaly.Mamaghani.ea2024} applied to the local representative $\Phi_\alpha$ on $\Pi(C_u)_\alpha\times \mathbb{R}^m$, the local closed-loop multifunction $\Phi_{\mathrm{cl},\alpha}(x)=\Phi_\alpha(x,\kappa^*_\alpha(x))$ is osc, locally bounded, nonempty and convex on $C_{\kappa,\alpha}$. 
Since this holds for every $\alpha\in A$, Lemma \ref{lem:Local_invariants} lifts these properties back intrinsically: $\Phi_{\mathrm{cl}}$ is osc, locally precompact relative to $C_\kappa$, nonempty, and fiberwise convex on $C_\kappa$. 
Additionally, by \cite[Lem.~2]{Isaly.Mamaghani.ea2024} applied to the local representative $C_{u,\alpha}$, the closed-loop $C_{\kappa,\alpha}$ is closed. Since $\phi_\alpha^{-1}$ is 
a homeomorphism, $C_\kappa \cap \mathcal{U}|_{U_\alpha}$ is closed in 
$\mathcal{U}|_{U_\alpha}$ for every $\alpha \in A$. Since 
$\{U_\alpha\}_{\alpha \in A}$ covers $\mathcal{M}$ and closedness is local, $C_\kappa$ is closed in $\mathcal{U}$.
\end{proof}

\subsection{Sheaf Compatibility Conditions}
\label{sec: Sheaf Compatibility Conditions}

From Section \ref{sec: Fiberwise_Regularity}, we recover fiberwise regularity conditions on the multifunction $K_{\mathrm{CBF}}(x)$.
Under these conditions, 
compatible local controllers that glue determine a unique global section.   This is formalized in Definition \ref{def: Sheaf_CBF} and Theorem \ref{thm: Global_section}, which establishes the existence of a global continuous feedback law obtained by gluing local solutions.

\begin{definition}
\label{def: Sheaf_CBF}
    Let $\{(U_\alpha,\phi_\alpha)\}_{\alpha \in A}$ be a bundle atlas and cover of $\Pi(C_u)$ for $\alpha\in A$. A continuously differentiable CBF candidate $B:\mathcal{M}\rightarrow\mathbb{R}^v$ is a \textit{Sheaf Control Barrier Function} (Sheaf-CBF) for \eqref{eq:System_1} and $\mathcal{S}$ on $\mathcal{P}$ if:
    \begin{enumerate}
        \item $B$ is a CBF in the sense of Definition \ref{def:CBF_on_a_bundle} such that $K_{\mathrm{CBF}}(x)\neq\emptyset$ for all $x\in\mathcal{P}\subset\Pi(C_u);$
        \item The conditions of Theorem \ref{thm:Local_CBF} hold on each $U_\alpha$, so that $\kappa^*_\alpha \in \Gamma(U_\alpha,\mathcal{U})$ is single-valued and continuous with $\kappa^*_\alpha(x)\in K_{\mathrm{CBF},\alpha}(x)$ for all $x\in \Pi(C_u)_\alpha;$ 
        \item The local sections are compatible on overlapping trivializations such that, for all $\alpha,\beta \in A$ and all $x\in U_\alpha \cap U_\beta$, the restriction condition in \eqref{eq: Restriction Compatibility condition} is satisfied for the local sections of $\mathcal{U}$, and equivalently in local coordinates, the gluing condition is satisfied:
        \begin{equation}
                    \kappa^*_{\alpha}(x) = g_{\alpha\beta}(x)\, \kappa^*_{\beta}(x), \quad \forall x \in U_\alpha \cap U_\beta.
                    \label{eq: Gluing Compatibility condition}
        \end{equation}
    \end{enumerate}
\end{definition}

\begin{theorem}
\label{thm: Global_section}
    Let $B:\mathcal{M}\rightarrow\mathbb{R}^v$ be a Sheaf-CBF for \eqref{eq:System_1} and let $\{(U_\alpha,\phi_\alpha)\}_{\alpha \in A}$ be a bundle atlas covering $\Pi(C_u)$. A collection of compatible local sections $\{\kappa^*_\alpha \} $ defined on the coverage uniquely determines a global section 
    \begin{equation}
        \kappa^*\in\ \mathcal{K}_{\mathrm{CBF}}(\Pi(C_u))\subseteq \Gamma(\Pi(C_u),\mathcal{U}).
    \end{equation}
\end{theorem}

\begin{proof}
For the bundle atlas $\{(U_\alpha, \phi_\alpha)\}_{\alpha \in A}$, $x \in \Pi(C_u)$ exists on a local trivialization. 
Let $\kappa^*_\alpha(x)\in U_\alpha$ and $\kappa^*_\beta(x) \in U_\beta$ be local controllers for any $\alpha,\beta\in A$. Then on the overlap $x \in U_\alpha \cap U_\beta$, the transition function $g_{\alpha\beta}(x) \in \mathrm{GL}(m,\mathbb{R})$ satisfies $\phi_\alpha \circ \phi_\beta^{-1}(x,u) = 
(x, g_{\alpha\beta}(x)u)$, so the compatibility condition in \eqref{eq: Gluing Compatibility condition} is satisfied. Therefore, by the restriction condition in \eqref{eq: Restriction Compatibility condition}, $\kappa^*_\alpha(x)$ and $\kappa^*_\beta(x)$ are local representatives of the same element $\kappa^*(x)$ which confirms that the local sections $\kappa^*_\alpha(x) \in \Gamma(U_\alpha,\mathcal{U})$ and $\kappa^*_\beta(x) \in \Gamma(U_\beta,\mathcal{U})$ are compatible on overlapping
trivializations. By Theorem \ref{thm:Local_CBF}, each $\kappa^*_\alpha$ is single-valued and continuous on $\Pi(C_u)_\alpha$. By Lemma \ref{lem:Local_invariants}, continuity for $\kappa^*_\alpha(x)$ is chart-invariant under  $g_{\alpha\beta}(x)$, so $\kappa^*(x)$ is continuous. Since the atlas covers $\Pi(C_u)$, $\kappa^*$ is continuous on $\Pi(C_u)$. Therefore, $\pi \circ \kappa^*(x) = x$ holds on all of $\Pi(C_u)$, confirming 
$\kappa^* \in \Gamma(\Pi(C_u), \mathcal{U})$. By Definition \ref{def: Sheaf_CBF} and Lemma \ref{lem:Local_invariants}, $K_{\mathrm{CBF},\alpha}$ is chart-invariant and therefore $\kappa^*(x)\in K_\mathrm{CBF}(x)$ for all $x\in\Pi(C_u)$. Therefore, by \eqref{eq: sheaf_KCBF}, $\kappa^*\in\ \mathcal{K}_{\mathrm{CBF}}(\Pi(C_u))\subseteq \Gamma(\Pi(C_u),\mathcal{U})$.
\end{proof}

The global section $\kappa^*$ obtained above provides a continuous feedback law defined intrinsically on the bundle. The following theorem utilizes the global continuous feedback law in Theorem \ref{thm: Global_section} to guarantee forward invariance. 

\begin{theorem}
\label{thm: Sheaf_forward_inariance}
Let Assumption \ref{asm: Differential_inclusion_Axton} hold for the control system in \eqref{eq:System_1}, let $B:\mathcal{M}\rightarrow\mathbb{R}^v$ be a Sheaf-CBF for $(\Phi,C_u)$ defining a compact safe set $\mathcal{S}\subset\Pi(C_u)$ on $\mathcal{P} \subset\Pi(C_u)$ with respect to $\gamma:\Pi(C_u) \rightarrow \mathbb{R}^v$. Suppose the local controllers $\kappa^*_\alpha$ for $\alpha\in A$ satisfy the sheaf compatibility conditions in \eqref{eq: Restriction Compatibility condition} and \eqref{eq: Gluing Compatibility condition}. Then $\kappa^*$ is single-valued and continuous, is the unique global section, and $\mathcal{S}$ is forward invariant under \eqref{eq: Closed_Loop_System}. 
\end{theorem}

\begin{proof}
By Theorem \ref{thm: Global_section}, the collection $\{\kappa^*_\alpha\}_{\alpha \in A}$ satisfying \eqref{eq: Gluing Compatibility condition} uniquely determines the global section $\kappa^*\in\ \mathcal{K}_{\mathrm{CBF}}(\Pi(C_u))$ that is single-valued and continuous on $\Pi(C_u)$. Hence, by \eqref{eq: sheaf_KCBF}, $\kappa^*(x)\in K_\mathrm{CBF}(x)$ for all $x\in\Pi(C_u)$. Therefore, $\kappa^{*}(x(t))\in{K}_\mathrm{CBF}(x(t))$ for all $t\in \mathrm{dom}(x)$. Satisfaction of Assumption \ref{asm: Differential_inclusion_Axton} then ensures forward pre-invariance of $\mathcal{S}$ under the closed-loop dynamics in \eqref{eq: Closed_Loop_System}.  Maximal solutions of continuous vector fields are either complete or leave every compact subset of $\mathcal{M}$  in finite time \cite[Prop.~2]{Maghenem.Sanfelice2021}. By the forward pre-invariance of $\mathcal{S}$, every solution initialized in $\mathcal{S}$ remains in $\mathcal{S}$ for all $t\in\mathrm{dom}(x)$. Since $\mathcal{S}$ is also compact, such a solution cannot leave every compact subset of $\mathcal{M}$ in finite time \cite[Lem. 9.19]{Lee2012}. Therefore, all maximal solutions are complete and $\mathcal{S}$ is forward invariant by Definition \ref{def: Forward-Invariance-Euclidean}. 
\end{proof}

\section{Simulated Results}
\label{sec: Simulated Results}
In this example, we consider the M{\"o}bius strip as the state manifold $\mathcal{M}$ and a control bundle $(\mathcal{U},\pi,\mathcal{M})\cong(T\mathcal{M},\pi_{T\mathcal{M}},\mathcal{M})$, where the half-twist renders the control bundle globally nontrivial and sheaf conditions in~\eqref{eq: Restriction Compatibility condition} and~\eqref{eq: Gluing Compatibility condition} nontrivial. 
The dynamics are given as $\dot{x}=F(x)+u_x$ for $x\in\mathcal{M}$ and $u_x\in \mathcal{U}_x \cong T_x\mathcal{M}$. 
The base coordinate is denoted by $\theta\in S^1$ and the fiber coordinate is denoted by $\ell\in\mathbb{R}$. 
The control objective is to track a desired trajectory $x_d(t)\in\mathcal{M}$ while remaining within the safe set $\mathcal{S}\subset\mathcal{M}$. 

We consider an open cover composed of two chart domains $U_\alpha\triangleq S^1\backslash \{ \theta=\pi \}$ and $U_\beta\triangleq S^1\backslash \{ \theta = 0 \}$, with transition function between charts $g_{\alpha\beta}\triangleq\mathtt{diag}\{1,-1\}$ on the arc $\pi < \theta \leq 2\pi$ and $g_{\alpha\beta}\triangleq\mathtt{diag}\{1,1\}$ otherwise, encoding the half-twist of the M{\"o}bius strip.
A transition between charts occurs when the agent enters a region of $\pm 0.05\ \mathrm{rad}$ about the obstruction on the chart; e.g., if the agent's current local chart is $U_\alpha$ and the agent enters $\pi - 0.05 \leq \theta\leq \pi + 0.05$, then the agent applies the transition function and enters chart $U_\beta$. A similar transition occurs for active chart $U_\beta$ and the region $-0.05 \leq \theta \leq 0.05$. 

\begin{figure}
    \centering
    \includegraphics[width=0.9\linewidth]{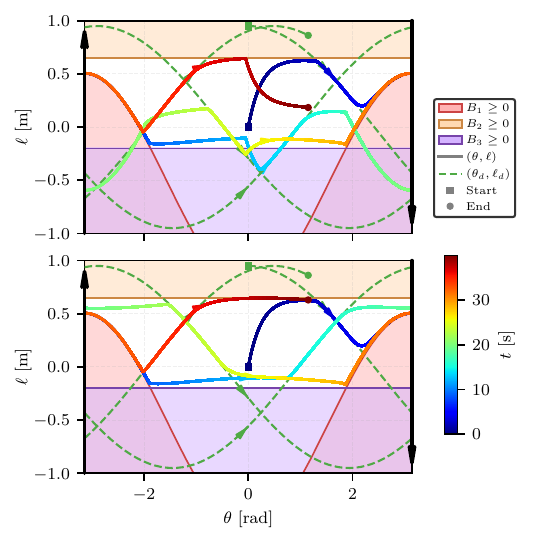}
    \caption{The top panel depicts a CBF which does not satisfy the gluing condition in~\eqref{eq: Gluing Compatibility condition}, causing a violation of safety. The bottom panel depicts the Sheaf-CBF regulating safety. The state progression on $\mathcal{M}$ is depicted on $U_\alpha$. 
    The progression of time is depicted through the changing color map indicated on the legend.
    The shaded regions indicate unsafe portions of the state space. Arrows along each trajectory indicate the direction of time evolution. The black arrows on the panel edges indicate the half-twist of the M{\"o}bius strip. When the arrowed edges are aligned and glued, the nonorientable surface is constructed.}
    \label{fig:Flat-Chart-Comparision}
\end{figure}

We express the following quantities in chart-local coordinates for a chart domain denoted as $c\in\{ \alpha,\beta\}$. 
The components of the desired trajectory are given as $\theta_d (t)\triangleq \mathtt{mod} (10t,2\pi)$ and $\ell_d(t)\triangleq 0.95\cos (0.325t - 2\pi)$. 
The components of the tracking error are $e_\theta\triangleq \mathtt{atan2}(\sin(\theta-\theta_d),\cos(\theta-\theta_d))$, where $\mathtt{atan2}$ is the $2$-argument arctangent, and $e_{\ell}\triangleq \ell - \ell_d$. 
The nominal controller is given as $u_{\text{nom}}=[-10 e_\theta + \dot{\theta}_d ,\, -10e_\ell + \dot{\ell}_d ]^\top$. 
The agent is initialized in chart $\alpha$ with $\theta (0)=\ell(0)=0$.
The simulation is conducted for a duration of $\mathtt{tf}=40\ \mathrm{s}$ with a step size of $\mathtt{dt}=0.001\ \mathrm{s}$.

The barrier function is $B^c(x)\triangleq [\cos(\frac{2\pi}{3})-\cos (\theta)-\ell,\ell -0.65,-\ell-0.2]^\top$. 
The first element of $B^c$ encodes a constraint on $\theta$ which can be circumvented by a sufficiently large value of $\ell$. 
The remaining elements of $B^c$ encode constraints on the fiber such that $-0.2\leq \ell (t)\leq 0.65$ for all $t\in\mathbb{R}_{\geq 0}$. 
Each element of $B^c$ has a performance function $\gamma_i^c(x)=2B_i^c(x)$ and $C_u\triangleq \{ u\in\mathbb{R}^2\ \vert\ \lVert u \rVert \leq 20 \}$. 
The safe controller at $x$ is the solution to $\kappa_c^\ast (x)\triangleq \arg\min_{u\in K_{\mathrm{CBF},c}(x)}\lVert u - u_{\text{nom}} \rVert^2$, where $K_{\mathrm{CBF},c}(x)$ is defined in~\eqref{eq: K_CBF}. 
The closed-loop dynamics are expressed as $[ \dot{\theta} ,\; \dot{\ell} ]^\top=\, [ (\kappa^\ast_c (x))_\theta, \; 0.1 \cos(\theta) + (\kappa_c^\ast (x))_\ell ]^\top$, where $(\kappa_c^\ast (x))_\theta,(\kappa_c^\ast (x))_\ell$ denote the $\theta$ and $\ell$ components of $\kappa_c^\ast$, respectively. 

Figure~\ref{fig:Flat-Chart-Comparision} depicts the evolution of two CBF-based controllers performing the trajectory tracking task on $\mathcal{M}$ with respect to $U_\alpha$.
The controller in the bottom panel of Figure~\ref{fig:Flat-Chart-Comparision} applies the transition map $g_{\alpha\beta}$ to its CBF at appropriate transitions between charts, while the controller in the top panel ignores this transition 
for all $t\in [0,40]\ \mathrm{s}$. The CBF that is not updated according to the transition function is equivalent to treating the state space $\mathcal{M}$ as the globally trivial domain $S^1\times\mathbb{R}$. Consequently, after traversing about $\theta$ to the identified point at $0\sim2\pi$, neglecting the transition functions and gluing conditions in~\eqref{eq: Gluing Compatibility condition} produces incompatible local coordinates that persist until the agent's angular state wraps again to $0\sim 2\pi$, resulting in safety violations. The Sheaf-CBF respects the transitions, and thereby enables the agent to remain in the safe set for the simulation duration. 

\section{Conclusion}
This work introduces the first sheaf-theoretic framework for CBFs for constrained control differential inclusions on vector bundles. Intrinsic regularity conditions on the bundle were established fiberwise under local trivializations. Using local properties, the sheaf restriction and gluing conditions on overlapping trivializations ensured compatibility for assembling local controllers into a unique, continuous global section. With the additional sheaf-theoretic notions for local feedback, we construct the Sheaf-CBF and provide forward-invariance guarantees. Simulations on a nontrivial bundle demonstrate the necessity of the sheaf conditions across overlapping charts and successfully constrain the states within the intersecting safe region.

\section{Appendix}
\subsection{Proof for Lemma 1}
\begin{proof}
(i) Suppose $\Psi_\alpha$ is osc at $x$ and let $e_x \in \limsup_{z\rightarrow x} \Psi(z)$. Then there exists $z_i\rightarrow x$ and $e_{z_i}\rightarrow e_x$ in $E$ with $e_{z_i}\in\Psi(z_i)$. For sufficiently large $i$, $z_i \in U_\alpha$ so setting $v_i\triangleq \mathrm{pr}_2 \circ \phi_\alpha(e_{z_i})\in \Psi_\alpha(z_i)$ and $v\triangleq\mathrm{pr}_2 \circ \phi_\alpha(e_{x})$, by the continuity of $\phi_\alpha$ yields $v_i\rightarrow v$. Since $\Psi_\alpha$ is osc, $v\in\Psi_\alpha(x)$, and since $v=\mathrm{pr}_2 \circ \phi_\alpha(e_x)$, then $e_x=\phi^{-1}_\alpha (x,v)\in \Psi(x)$. Therefore $\limsup_{z\rightarrow x} \Psi(z)\subset\Psi(x)$ and $\Psi$ is osc at $x$. Now suppose $\Psi$ is osc at $x$ and let $v \in \limsup_{z\rightarrow x} \Psi_\alpha(z)$. Then there exists $z_i\rightarrow x$ and $v_{i}\rightarrow v$ in $\mathbb{R}^k$ with $v_{i}\in\Psi_\alpha(z_i)$. Setting $e_{z_i}\triangleq \phi^{-1}_\alpha (z_i,v_i)\in \Psi(z_i)$, by the continuity of $\phi^{-1}_\alpha$ yields $e_{z_i} \rightarrow e_x$ where $e_x = \phi^{-1}_\alpha(x,v)\in E_x$. Since $\Psi$ is osc, then $e_x\in \Psi(x)$, hence $v =\mathrm{pr}_2 \circ \phi_\alpha(e_{x}) \in \Psi_\alpha(x)$. Therefore $\limsup_{z\rightarrow x} \Psi_\alpha(z)\subset\Psi_\alpha(x)$ and $\Psi_\alpha$ is osc at $x$. 

Suppose $\Psi_\alpha$ is isc at $x$ and let $e_x\in\Psi(x)$ and $z_i \rightarrow x$ with $z_i\in U_\alpha$ for $i$ sufficiently large. With $v\triangleq \mathrm{pr}_2 \circ\phi_\alpha(e_x) \in\Psi_\alpha(x)$, by the isc of $\Psi_\alpha$, there exists $v_i\in \Psi_\alpha(z_i)$ with $v_i\rightarrow v$. Setting $e_{z_i} \triangleq \phi^{-1}_\alpha (z_i,v_i) \in \Psi(z_i)$, by the continuity of $\phi^{-1}_\alpha$ yields $e_{z_i} \rightarrow e_x$ where $e_x = \phi^{-1}_\alpha(x,v)\in \Psi(x)$. Since $z_i \to x$ was an arbitrary sequence, $e_x \in \liminf_{z \to x}\Psi(z)$. Since $e_x \in \Psi(x)$ was arbitrary, $\Psi(x) \subset \liminf_{z \to x}\Psi(z)$ and $\Psi$ is isc at $x$. Now suppose $\Psi$ is isc at $x$ and let $v\in\Psi_\alpha(x)$ and $z_i \rightarrow x$ with $z_i\in U_\alpha$. Setting $e_x\triangleq \phi^{-1}_\alpha (x,v)\in \Psi(x)$, by isc of $\Psi$, there exists $e_{z_i}\in \Psi(z_i)$ with $e_{z_i}\rightarrow e_x$. Setting $v_i\triangleq \mathrm{pr}_2 \circ \phi_\alpha(e_{z_i}) \in \Psi_\alpha(z_i)$, by the continuity of $\phi_\alpha$ yields $v_i \rightarrow v$ where $v=\mathrm{pr}_2 \circ \phi_\alpha(e_x)$. Since $z_i \to x$ was arbitrary, $v \in \liminf_{z \to x}\Psi_\alpha(z)$. Since $v \in \Psi_\alpha(x)$ was arbitrary, $\Psi_\alpha(x) \subset \liminf_{z \to x}\Psi_\alpha(z)$ and $\Psi_\alpha$ is isc at $x$.

The final direction shows chart independence regarding osc and isc. 
For osc, assume $\Psi_\alpha$ is osc at $x\in U_\alpha\cap U_\beta$ and let $v\in \limsup_{z\rightarrow x} \Psi_\beta(z)$. Then there exists $z_i\rightarrow x$ and $v_i\rightarrow v$ with $v_i\in\Psi_{\beta}(z_i)$. Setting $w_i\triangleq g_{\alpha\beta}(z_i)v_i$ yields $w_i\in\Psi_{\alpha}(z_i)$. The transition function $g_{\alpha\beta}$ is continuous, yielding $w_i\rightarrow g_{\alpha\beta}(x)v$. Hence, $g_{\alpha\beta}(x)v \in \limsup_{z\rightarrow x} \Psi_\alpha(z)$. Therefore $v\in\Psi_{\beta}(x)$ and $\limsup_{z\rightarrow x} \Psi_\beta(z)\subset\Psi_{\beta}(x)$. For isc, assume $\Psi_\alpha$ is isc at $x\in U_\alpha\cap U_\beta$ and let $v\in \Psi_\beta(x)$. Then setting $w\triangleq g_{\alpha\beta}(x)v \in \Psi_\alpha(x)$. For every sequence $z_i\rightarrow x$, there exists $w_i\in \Psi_{\alpha}(z_i)$ with $w_i\rightarrow w$. Setting $v_i\triangleq g_{\beta\alpha}(z_i)w_i$ and the continuity of $g_{\beta\alpha}$ yields $v_i\in \Psi_\beta(z_i)$ and $v_i\rightarrow g_{\beta\alpha}(x) w =v$. Thus, $v\in \mathrm{lim}\,\mathrm{inf}_{z\rightarrow x} \Psi_\beta(z)$, and therefore $\Psi_\beta(x) \subset \liminf_{z\rightarrow x} \Psi_\beta(z)$. The converse for osc and isc follows symmetrically. If both conditions hold, then $\Psi$ is continuous at $x$.  

(ii) Since $\left. \phi_\alpha \right|_{E_x}: E_x\rightarrow \{x\}\times \mathbb{R}^k$ is a linear isomorphism \cite[Lem. 10.6]{Lee2012}, by \cite[Prop. 2.21]{Rockafellar.Wets2009}$, \Psi(x)\subset E_x$ convex if and only if $\phi_\alpha(\Psi(x)) = \{x\}\times \Psi_\alpha(x)$ is convex in $\{x\}\times \mathbb{R}^k$. If $\Psi(x)\subset E_x$ convex, then for any $e_0,e_1\in\Psi(x)$, with $u_0=\phi_\alpha(e_0),\,u_1=\phi_\alpha(e_1)$, and $\tau\in[0,1]$, the linearity of $\phi_\alpha$ on a convex combination gives $(1-\tau)u_0 + \tau u_1 = \phi_\alpha((1-\tau)e_0+\tau e_1) \in \phi_\alpha(\Psi(x))=\Psi_\alpha(x)$. The converse follows by applying the identical argument to $\phi^{-1}$. Consequently, convexity in $\{x\}\times \Psi_\alpha(x)$ holds if and only if $\Psi_\alpha(x)\subset\mathbb{R}^k$ is convex. 
The transition function $g_{\alpha\beta}(x)\in\mathrm{GL}(k,\mathbb{R})$ is linear and by \cite[Prop. 2.21]{Rockafellar.Wets2009} preserves convexity on overlapping charts given $\Psi_\alpha(x)=g_{\alpha\beta}(x)\Psi_\beta(x)$.

(iii) Suppose $\Psi_\alpha$ is locally bounded at $x$ and let $V \in \mathcal{N}(x)$ with $\overline{V} \subset U_\alpha$ such that $\Psi_\alpha(V)$ is bounded in $\mathbb{R}^k$. By Heine-Borel, $\overline{\Psi_\alpha(V)}$ is compact in $\mathbb{R}^k$ and since $\mathcal{M}$ is locally compact \cite[Prop. 1.12]{Lee2012}, then $\overline{V}$ is compact in $U_\alpha$ and the product $\overline{V} \times \overline{\Psi_\alpha(V)}$ is compact in $U_\alpha\times \mathbb{R}^k$. Since $\phi^{-1}_\alpha$ is a homeomorphism, $W \triangleq \phi^{-1}_\alpha(\overline{V} \times \overline{\Psi_\alpha(V)})$ is compact in $\left. E \right|_{U_\alpha}$. For every $z\in V$ and every $e_{z}\in\Psi(z)$, there exists $\phi_\alpha(e_{z})=(z,v)$ with $v\in\Psi_\alpha(z)\subset\Psi_\alpha(V)$, hence $e_z\in W$. Therefore $\Psi(V)\subset W$. Since $E$ is Hausdorff, $W$ is closed in $E$, so $\overline{\Psi(V)} \subset \overline{W}$. A closed subset of a compact set is compact, hence $\overline{\Psi(V)}$ is compact and therefore $\Psi$ is locally precompact at $x$. 
Now suppose $\Psi$ is locally precompact at $x$ and let $V \in \mathcal{N}(x)$ with $V \subset U_\alpha$ such that $\overline{\Psi(V)}$ is compact in $\left. E \right|_{U_\alpha}$. Since $\phi_\alpha$ is a homeomorphism, $\phi_\alpha(\overline{\Psi(V)})$ is compact in $U_\alpha\times \mathbb{R}^k$. Since the projection is continuous, $\mathrm{pr}_2(\phi_\alpha(\overline{\Psi(V)}))$ is compact in $\mathbb{R}^k$. For every $z\in V$ and $e_z\in\Psi(z)$, $\phi_\alpha(e_z)=(z,v)$ with $v\in\Psi_\alpha(z)\subset \Psi_\alpha(V)$. Therefore, $\Psi_\alpha(V) \subset \mathrm{pr}_2(\phi_\alpha(\overline{\Psi(V)}))$, and lies in a compact set in $\mathbb{R}^k$. By Heine-Borel, $\Psi_\alpha(V)$ is bounded, and $\Psi_\alpha$ is locally bounded at $x$. Since local precompactness of $\Psi$ at $x$ is a property of the total space $E$, $\Psi_\alpha$ is locally bounded at $x$ if and only if $\Psi_\beta$ is locally bounded at $x$ 
when $x\in U_\alpha \cap U_\beta$.
\end{proof}

\bibliographystyle{ieeetr}
\bibliography{master}

@article{Isaly.Mamaghani.ea2024,
  url1 = { http://ncr.mae.ufl.edu/papers/tac24_4.pdf },
  volume = { 69 },
  pages = { 7326 - 7339 },
  number = { 11 },
  year = { 2024 },
  title = { On the Feasibility and Continuity of Feedback Controllers Defined by Multiple Control Barrier Functions },
  journal = { IEEE Trans. Autom. Control },
  author = { A. Isaly and M. Mamaghani and R. G. Sanfelice and W. E. Dixon },
}

@InProceedings{Ames.Grizzle.ea2014,
  Title                    = {Control barrier function based quadratic programs with application to adaptive cruise control},
  Author                   = {Ames, Aaron D and Grizzle, Jessy W and Tabuada, Paulo},
  Booktitle                = {IEEE Conf. Decis. Control},
  Year                     = {2014},
  Pages                    = {6271--6278}
}

@InProceedings{Ames.Coogan.ea2019,
  author    = {Aaron D. Ames and Samuel Coogan and Magnus Egerstedt and Gennaro Notomista and Koushil Sreenath and Paulo Tabuada},
  booktitle = {Proc. Eur. Control Conf.},
  title     = {Control barrier functions: Theory and applications},
  year      = {2019},
  pages     = {3420-3431},
}

@Book{Rockafellar.Wets2009,
  Title                    = {Variational analysis},
  Author                   = {Rockafellar, R Tyrrell and Wets, Roger J-B},
  Publisher                = {Springer Science \& Business Media},
  Year                     = {2009},
  Volume                   = {317}
}

@Article{Maghenem.Sanfelice2021,
  author    = {Maghenem, Mohamed and Sanfelice, Ricardo G},
  journal   = {Automatica},
  title     = {Sufficient conditions for forward invariance and contractivity in hybrid inclusions using barrier functions},
  year      = {2021},
  volume    = {124},
  publisher = {Elsevier},
}

@book{Lee2012,
  title = {Introduction to {{Smooth Manifolds}}},
  author = {Lee, John M.},
  year = 2012,
  series = {Graduate {{Texts}} in {{Mathematics}}},
  volume = {218},
  publisher = {Springer New York},
}

@book{Tu2011,
  title = {An {{Introduction}} to {{Manifolds}}},
  author = {Tu, Loring W.},
  year = 2011,
  publisher = {Springer New York},
}

@inproceedings{Gonzalez.Sweatland.ea2025,
  year = { 2025 },
  title = { Pointwise Optimal Feedback Laws for Hybrid Inclusions using Multiple Control Barrier Functions },
  booktitle = { Proc. IEEE Conf. Decis. Control },
  author = { C. Gonzalez and H. Sweatland and K. Currier and W. E. Dixon and R. Sanfelice },
}

@book{Bullo.Lewis2005,
  Author = {Bullo, Francesco and Lewis, Andrew D},
  Title = {Geometric control of mechanical systems: modeling,
           analysis, and design for simple mechanical control
           systems},
  Volume = {49},
  Year = {2005},
  Publisher = {Springer}
}

@Book{Goebel.Sanfelice.ea2012,
  author    = {Goebel, R. and Sanfelice, R. G. and Teel, A. R.},
  publisher = {Princeton University Press},
  title     = {Hybrid Dynamical Systems},
  year      = {2012},
}

@Book{Sanfelice2021,
  author    = {Ricardo G. Sanfelice},
  publisher = {Princeton University Press},
  title     = {Hybrid Feedback Control},
  year      = {2021},
}

@InProceedings{Wu.Sreenath2015,
  author    = {Guofan Wu and Koushil Sreenath},
  booktitle = {Proc. Am. Control Conf.},
  title     = {Safety-critical and constrained geometric control synthesis using control {L}yapunov and control barrier functions for systems evolving on manifolds},
  year      = {2015},
  pages     = {2038-2044},
}

@Article{deSa.Ong.ea2025,
  Author  = {de Sa, Massimiliano and Ong, Pio and Ames, Aaron D.},
  Title   = {From Bundles to Backstepping: Geometric
             Control Barrier Functions for
             Safety-Critical Control on Manifolds},
  Journal = {arXiv preprint arXiv:2510.20202},
  Year    = {2025}
}

@InProceedings{Jirwankar.Sanfelice2025a,
  Author    = {Jirwankar, P. P. and Sanfelice, R. G.},
  Title     = {Geometric Hybrid Dynamical Systems on
               Intrinsic Manifolds: Part {I} ---
               Framework and Hybrid Lyapunov Theorem},
  Booktitle = {Proc. IEEE Conf. Decis. Control},
  Year      = {2025},
  Pages     = {6593--6598}
}

@InProceedings{Jirwankar.Ochoa.ea2025b,
  Author    = {Jirwankar, P. P. and Ochoa, D. E. and Sanfelice, R. G.},
  Title     = {Geometric Hybrid Dynamical Systems on Intrinsic Manifolds: Part {II} --- Hybrid Invariance Principle},
  Booktitle = {Proc. IEEE Conf. Decis. Control},
  Year      = {2025},
  Pages     = {6599--6604}
}

@book{MacLane.Moerdijk.ea1994,
  title     = {Sheaves in {Geometry} and {Logic}: {A} First Introduction 
               to Topos Theory},
  author    = {Mac Lane, Saunders and Moerdijk, Ieke},
  year      = {1994},
  publisher = {Springer},
}

@book{Lee2024,
  author    = {John M. Lee},
  title     = {Introduction to Complex Manifolds},
  series    = {Graduate Studies in Mathematics},
  volume    = {244},
  publisher = {American Mathematical Society},
  year      = {2024},
}

@unpublished{TemporalCellularSheaf2026,
author = {Fallin, Brandon and Leal, Wilmer and Copeland, Austin and Bumpus, Benjamin and Fairbanks, James and Dixon, Warren E.},
title = {A Temporal Cellular Sheaf Framework for Multi-Agent Systems with Switching Topologies},
year = {2026},
note = {Unpublished}
}

@Article{LealNarratives2026,
title={Time-Varying Data as Sheaves: an Invitation to Narratives},
author={Wilmer Leal and Benjamin Merlin Bumpus and Jana K. Nickel and Johan Garc\'ia and James Fairbanks and Warren Dixon},
year={2026},
Journal = {arXiv preprint arXiv:2609.09056},
url={https://arxiv.org/abs/2609.09056},
}

\end{document}